\documentclass[11pt]{article}
\usepackage{graphicx}

\usepackage{amssymb}
\usepackage{amsmath}
\usepackage{amsfonts}

\newtheorem{theorem}{Theorem}[section]
\newtheorem{proposition}[theorem]{Proposition}
\newtheorem{corollary}[theorem]{Corollary}
\newtheorem{definition}[theorem]{Definition}
\newtheorem{lemma}[theorem]{Lemma}
\newtheorem{remark}[theorem]{Remark}

\newtheorem{conjecture}{Conjecture}[section]
\newtheorem{example}[theorem]{Example}

\newcommand{\qed}{\hfill $\ \Box $\medskip }
\newenvironment{proof}[1][Proof]{\noindent\textbf{#1.} }%
{\qed}

\title{Methods of Class Field Theory to Separate Logics over Finite Residue Classes and Circuit Complexity}

\author{Argimiro  Arratia\footnote{Supported  by  MICINN  project  TIN2011-27479-C04-03 (BASMATI), 
MINECO project TIN2014-57226-P (APCOM)
and Gen. Cat. 
project SGR2014-890 (MACDA).}\\  
Dept. of Computer Science\\
Universitat Polit\`ecnica de Catalunya (Barcelona Tech), Spain  %
\\
argimiro@cs.upc.edu
\and 
 Carlos E. Ortiz \\ 
Dept. of Computer Science
and Mathematics, \\
Arcadia University, 
U.S.A.\\
ortiz@arcadia.edu 
}

\date{November 2014}

\begin{document}

\maketitle

\begin{abstract}
Separations among  the first order logic   ${\cal R}ing(0,+,*)$  of
finite residue class rings, its extensions with genera\-lized quantifiers, and in the 
presence of a built-in order are shown, using algebraic methods from class field theory.
These methods include classification of spectra of sentences over finite residue classes
as systems of congruences, and the study of 
 their $h$-densities over the set of all prime numbers, for various functions
  $h$ on the natural numbers. 
Over ordered structures the logic of finite residue class rings and extensions are known to capture DLOGTIME-uniform circuit complexity classes ranging from $AC^0$ to $TC^0$.
Separating these circuit complexity classes is directly 
related to  classifying the  $h$-density of spectra of sentences in the corresponding logics of finite residue classes.   
 We further give general conditions under which 
 a logic over the finite residue class rings 
 has a sentence whose spectrum has no $h$-density.  
 One application of this result  
 is that in   
 ${\cal R}ing(0,+,*,<) + M$, the logic of finite residue class rings with built-in order and extended with
 the majority quantifier  $M$, there are sentences whose spectrum have no exponential density.


\end{abstract}

\section{Introduction}

A first order logic for finite residue class rings, denoted   
as ${\cal R}ing(0,+,*)$ with $+$ and $*$ being modular addition and multiplication, and extensions of this logic 
 with certain  generalized quantifiers,  
 were shown by us in   \cite{AACO13} to
 coincide with the first order logic, and corresponding extensions,  for 
 the standard finite models with arithmetic operations as 
 considered in \cite{Immbook}. Therefore, from the 
 descriptive complexity perspective, the
 computational complexity classes that can be  described 
 with these logics for finite residue class rings are the $DLOGTIME$-uniform circuit complexity classes,
consisting of circuits of  constant depth and polynomial size
for which a description can be efficiently computed from the size of their inputs.
 
The perspective of finite residue class rings as instances of problems in these circuit complexity classes allow us to leverage the algebraic machinery proper of finite rings and fields of integer polynomials, algebraic number  theory, 
and in general tools from class field theory   to study 
expressibility problems in these logics.
In \cite{AACO13} we gave a brief account of the use of some of these algebraic tools to 
distinguish   the expressive power of the logic  of
finite residue class rings from its extensions with generalized quantifiers and in the 
addition of a built-in order. 
In this paper we review those tools in more detail and expand our methods and results.  

The algebraic methodology  
stems from the class field notion of spectra of polynomials 
over finite fields adapted to sets of sentences in the logics of finite residue classes.
The spectrum of a sentence $\theta$ is the collection
of primes $p$ such that $\theta$ holds in the residue class mod $p$. This spectra for first order sentences was extensively studied by Ax in \cite{ax2,ax} in connection with the decidability of the elementary theory of finite fields.  

Working with the concept of spectra of sentences, our
  strategy  
 to show separation among two logics, say $\cal L$ and $\cal L'$, over finite structures, can be described as follows:
\begin{quote}
Prove that the spectra of all sentences in $\cal L$ have property $P$, but there exists a 
sentence in $\cal L'$ whose spectrum does not have property $P$.
\end{quote}
From the algebraic machinery of class field theory we exploit  as candidates
for property $P$ the following:
\begin{itemize}
\item certain characterization of spectra of polynomial congruences as specific 
sets of congruent integers;  
\item different notions of density of sets of primes, in particular, the natural and the exponential density.  
\end{itemize}

We illustrate the power of the methods from class field theory in solving definability 
problems in descriptive complexity in many ways. In fact, to illustrate the richness of methods that can be obtained from the proposed algebraic setting,  we will provide three different proofs of the separation of ${\cal R}ing(0,+,*)$ from its ordered extension 
${\cal R}ing(0,+,*, <)$ by  different methods, one of combinatorial nature and the other two analytical around the concept of density.

The paper is organized as follows.  After some necessary preliminaries in circuit complexity and descriptive complexity (Section \ref{sec:2}), 
we define in Section \ref{sec:FinRing}  
 the logic of finite residue class rings, and state without proof 
 our result from \cite{AACO13} on capturing  circuit complexity classes by these logics.
Then in Section \ref{S:RingSp}  
 we review the notion of the 
 prime spectrum of a sentence,  its characterization as the spectrum of polynomial equations and the properties of the Boolean algebra 
 of this spectra, derived from the seminal results by Ax \cite{ax2,ax}.
 In Section  \ref{sec:5} 
we exhibit  a characterization of the spectra of sentences 
in ${\cal R}ing(0,+,*)$ in terms of sets of integer congruences, using a result of 
Lagarias 
\cite{Lag83} on spectra of polynomial equations.   
This tool allow us to show that 
   ${\cal R}ing(0,+,*)$ and  its extensions with modular quantifiers, namely
   ${\cal R}ing(0,+,*) + MOD(q)$ for each natural $q>2$,
   have different expressive power. 
   We also show how to apply this number-theoretic  tool to differentiate the expressive power of 
   ${\cal R}ing(0,+,*)$  and 
   its ordered counterpart ${\cal R}ing(0,+,*,<)$. 
In Section \ref{sec:6} we introduce the general measure of $h$-density for subsets of prime numbers, and the particular cases of interest, namely, the natural 
 and exponential  densities. Here we show that the natural density of the spectra of sentences in the 
logic ${\cal R}ing(0,+,*)$ (over unordered structures) always exists, but differs in an essential way in their possible values from the spectra  of sentences in  the extension with built-in order. 
We give in Section  \ref{sec:7} conditions based solely on the function $h$ under which we can construct sets of primes with no $h$-density, and use this tool to show, in Section \ref{sec:8}, two  results on different densities:  One showing the existence of a spectrum of a sentence in ${\cal R}ing(0,+,*,<)$ with no natural density (although this spectrum does have exponential density); the other showing the existence of a spectrum of a sentence in 
${\cal R}ing(0,+,*,<)+MOD+M$ with no exponential density. 
We remark that showing the 
spectra of sentences of ${\cal R}ing(0,+,*,<)$ have all exponential density (all the examples we know do verify this), will yield another way of  separating  circuit classes $AC^0$ and $TC^0$ in the DLOGTIME-uniform setting. 
Section \ref{sec:9} contains our final remarks and conclusions.

\section{Preliminaries}\label{sec:2}
\subsection{Circuit complexity}
We are interested in the uniform version of the circuit complexity classes ranging 
from $AC^0$ to $TC^0$. 
Recall that  $AC^{0}$
is  the class of languages accepted by polynomial size, constant
depth circuits with NOT gates and unbounded fan-in AND and OR gates.
Extending $AC^0$ circuits with unbounded $MOD_q$ gates, for a fixed integer $q >1$,  one obtains the class
$ACC(q)$. For each integer $q > 1$, a $MOD_q$ gate reads its boolean input and 
returns a 1 if the sum of the input bits is divisible by $q$ (i.e. the sum is equal to 0 mod $q$); otherwise it returns 0. Putting together all the $ACC(q)$ classes we get the class $ACC$, that is, 
$\displaystyle ACC = \bigcup\limits_{q > 1} ACC(q) .$
On the other hand, extending $AC^0$ circuits with unbounded $MAJ$ gates
one obtains the class $TC^0$. A $MAJ$ (or {\em majority}) gate returns a 1 if the sum of $n$ bits given as input is greater or equal to $n/2$; otherwise it returns a 0.
The languages decided with the use of $MOD_q$ gates  can be decided by 
using $MAJ$ gates instead. This fact, together with all given definitions of these circuit classes, give us for all $q > 1$
(cf. \cite{BS90,Immbook})
$$AC^0 \subseteq ACC(q) \subseteq ACC \subseteq TC^0 .$$

A circuit family is uniform if a description 
of each circuit can be computed efficiently from the size of the input; otherwise it is non-uniform.
The uniformity condition is crucial to relate time and space complexity with size and depth (see \cite{Bor77}), the measures for circuit complexity, since it can be shown that in the non 
uniform setting there are sets with trivial circuit complexity that are not  recursive. 
Various lower bounds are known for the non uniform versions of the aforementioned circuit complexity classes (see the survey \cite{BS90}). 
Two by now famous results,  one given by Furst, Saxe and Sipser in \cite{FSS84} and  
independently by Ajtai in \cite{Aj83}, proves  the existence of languages decided with  $MOD_q$ gates that can not be decided in $AC^0$. Therefore, $AC^0 \not= ACC(q)$.
 The other given by Smolensky in \cite{Smo87},  proves that
 if $p$ and $q$ are distinct primes then $ACC(p) \not= ACC(q)$.
 This implies that  $MAJ$ gates are more powerful than  $MOD_p$ gates, for single prime $p$. Nonetheless, it is yet unknown if compositions of  different $MOD$ gates are sufficient for deciding all languages that are decided by $MAJ$ gates; that is, it is unknown if $ACC$ coincides or not with $TC^0$.
A major challenge is to 
find new methods that work for 
showing lower bounds within uniform circuits.

\subsection{Descriptive complexity}\label{ssec:DC}
Our approach to circuit complexity is through finite model theory, and as a consequence
we are working with circuit classes that are DLOGTIME-uniform, for as it has been 
shown in \cite{BIS90}, the languages in DLOGTIME-uniform circuit classes $AC^0$
to $TC^0$ are definable in first order logic with built-in arithmetic predicates and some generalized quantifiers.
This works as follows. Consider first  the basic logic $\mbox{FO}(\le,\oplus, \otimes)$, which is first order logic with  built-in order relation $\le$,
and two ternary predicates $\oplus$ and $\otimes$. In a  finite model for this logic, denoted here as
$\mathcal{A}_{m}$  ($m$ is the cardinality of the model, and its universe
 is $|\mathcal{A}_{m}| = \{0,1,\ldots, m-1\}$),  the interpretation of $\le$ on 
 $\mathcal{A}_{m}$ is as a total ordering on $|\mathcal{A}_{m}|$, 
and the interpretations of  $\oplus$ and $\otimes$ are  as truncated addition and multiplication (e.g. any pair of elements that add up -or multiplies- to a quantity greater than $m$ is not a defined triple).
Consider further  the following generalized quantifiers:
\begin{itemize}
\item[$(G1)$] {\em Modular} quantifiers, $\exists^{(r,q)}$, which for integers
$r$ and $q$, with $0 \le r < q$, and  formula $\phi(\overline{x},z)$, the
quantified formula $\exists^{(r ,q)} z \phi(\overline{a},z)$ holds in $\mathcal{A}_{m}$ if and only if the
number of values for $z$ that makes $\phi(\overline{a},z)$ true is equal to $r$ modulo $q$.
\item[$(G2)$] {\em Majority } quantifier, $M$, which for a formula 
 $\phi (\overline{x},z)$,  $(Mz)\phi (\overline{a},z)$ holds in $\mathcal{A}_{m}$
 if and only if $\phi (\overline{a},z)$ is
true for more than half of the possible values for $z$. 
\end{itemize}
Let  $\mbox{FO}(\le,\oplus, \otimes) + MOD(q)$, for a fixed integer $q>1$, be the logic
$\mbox{FO}(\le,\oplus, \otimes)$ extended with modular quantifiers with moduli $q$; that is, the set of first
order formulas as before plus the quantifiers $\exists^{(r,q)}$ with 
$0 \le r < q$.\\
Let
$\displaystyle \mbox{FO}(\le,\oplus, \otimes) + MOD = \bigcup\limits_{q > 1}
(\mbox{FO}(\le,\oplus, \otimes) + MOD(q)),$ and 
$\mbox{FO}(\le,\oplus, \otimes)+ M$ the logic extended 
with the  majority quantifier  $M$.
Barrington et al proved 
\begin{theorem}[\cite{BIS90}]\label{bisthm}
The languages in DLOGTIME-uniform class $\cal C$ are exactly those definable in the logic $\cal L$, 
where $\cal C$ is $AC^0$, $ACC(q)$, $ACC$ or $TC^0$, and $\cal L$ is 
$\mbox{FO}(\le,\oplus, \otimes)$, $\mbox{FO}(\le,\oplus, \otimes) + MOD(q)$, 
$\mbox{FO}(\le,\oplus, \otimes)+ MOD$ or $\mbox{FO}(\le,\oplus, \otimes) + M$, respectively. \qed
\end{theorem}

We will also need an alternative logical description of $TC^0$, namely, via
FO(COUNT) the first-order logic over structures with numbers and counting
quantifiers (refer to \cite[\S 12.3]{Immbook} for details).
This is  first--order logic interpreted over two sorted structures consisting of a 
standard structure ${\cal A}$ over a vocabulary $\tau$,   a numeric domain which is an 
initial segment of the naturals of length the size of  $\cal A$,
and numeric predicates. The syntax is extended with formulae of the form
$(\exists i x) \phi(x)$, 
 with $i$ ranging over the numeric domain and $x$ over $\cal A$, and its 
 meaning is 
 $|\{a\in {\cal A}: \phi(a)\}| \ge i .$
Moreover, the quantifier $(\exists ! i x)$ is used to denote that there exists exactly
$i$ $x$:
\[ (\exists ! i x) \phi(x) := (\exists i x) \phi(x) \wedge \neg (\exists i+1 x ) \phi(x)
\]
 We can also have quantifiers bounding the numeric elements and acting on numeric predicates. So, for example, in a vocabulary for graphs $\tau = \{E\}$, one 
 can express that a vertex $x$ has odd degree by the formula in FO(COUNT):
 \[ ODD(x) := (\exists i)(\forall j)((j+j\not= i) \wedge (\exists ! i y)E(x,y)) \]
 
 We have the following fact (see \cite[Prop. 12.16]{Immbook}):
 \begin{theorem}\label{focount2tc}
 Over ordered structures,  FO(COUNT) = $TC^0$. \qed
 \end{theorem}
 
%
%
%

\section{The logic of finite residue class rings and uniform circuit complexity classes}\label{sec:FinRing}
We use $\mathbb{Z}$ to denote the integers, $\mathbb{R}$ for the reals
 and $\mathbb{P}$ to denote the set of prime numbers.
For integers $a$, $b$ and $d$,  $b \equiv_d a$   denotes that $b$ is congruent to $a$ modulo $d$, and
$(a,b)$ stands for the greatest common divisor of $a$, $b$.
For each $m \in \mathbb{Z}$, $m>0$,  we denote by  $\mathbb{Z}_m$ the finite residue class
ring of $m$ elements.  As an algebraic structure $\mathbb{Z}_m$ consists of a set of elements 
$\{0,1, \ldots, m-1\}$, and
two binary functions $+$ and $*$ which corresponds to addition 
and multiplication modulo $m$, respectively.

\begin{definition}\label{def-ring}
By ${\cal R}ing(0,+,*)$ we 
denote the logic of finite residue class rings. This is the set
of first order sentences over the  built-in predicates
$\{0,+,*\}$, where $0$ is a constant symbol,  $+$ and $*$ are binary function symbols. The models of ${\cal R}ing(0,+,*)$ are 
the finite residue class rings $\mathbb{Z}_m$, and in each $\mathbb{Z}_m$, the $0$ is always 
 interpreted as the $0$-th residue class (mod $m$), and $+$ and $*$ are addition 
and multiplication modulo $m$.

By ${\cal R}ing(0,+,*,<)$ we denote the logic  ${\cal R}ing(0,+,*)$ further extended with an additional (built-in) order relation $<$. In this extension each finite ring $\mathbb{Z}_m$ is endowed with an order of its residue classes, given by the natural ordering of the representatives of each class from 
$\{0,1, \ldots, m-1\}$. Also, in this case, the constant $0$ represents the first element in this order.
\end{definition}

We can further extend ${\cal R}ing(0,+,*)$ or ${\cal R}ing(0,+,*,<)$ with modular quantifiers and the majority quantifier.

\begin{definition} For every integer $q > 0$, we denote by
${\cal R}ing(0,+,*) + MOD(q)$  and ${\cal R}ing(0,+,*,<) + MOD(q)$
 the extensions of the logics  
 ${\cal R}ing(0, +, *)$ and ${\cal R}ing(0, +, *, <) $ 
 obtained by the additional requirement that these logics be closed, for 
every $r=0,1, \ldots, q-1$,
 for  the quantifiers $\exists^{(r,q)}x$, interpreted 
 in $\mathbb{Z}_{m}$ as in $(G1)$ of Section \ref{ssec:DC}.\\
We define  ${\cal R}ing(0,+,*)+MOD={\cal R}ing(0,+,*) + \bigcup_{q>0}MOD(q)$ and 
${\cal R}ing(0,+,*,<)$ $+ MOD={\cal R}ing(0,+,*,<) + \bigcup_{q>0} MOD(q)$.
Finally, we denote by
${\cal R}ing(0,+,*)+MOD+M$ and ${\cal R}ing(0,+,*<)+MOD+M$ the extensions of the logic ${\cal R}ing(0,+,*)+MOD$ and 
${\cal R}ing(0,+,*,<)+MOD$ obtained by the additional requirement that these logics be closed for the majority quantifier $Mz$, interpreted 
in $\mathbb{Z}_{m}$ as in $(G2)$ of Section \ref{ssec:DC}.
  \end{definition}

 In the presence of a built-in order relation  it is logically indistinct to work  with the standard 
finite models ${\cal A}_m$ or with the finite residue class rings $\mathbb{Z}_m$.   
This is the contents of the following theorem whose 
 proof  can be found in \cite{AACO13}.

%

\begin{theorem}\label{ring=fo:order}
For every formula $\phi(x_1, \ldots, x_k)$ of $\mbox{FO}(\le,\oplus, \otimes)$,
 there exists a formula $\Phi(x_1, \ldots, x_k)$ of ${\cal R}ing(0,+,*,<)$
such that for every finite
structure $\mathcal{A}_{m}$ and integers  $a_{1},\ldots,a_{k}<m$, 
$$\mathcal{A}_{m}\models\phi(a_{1},\ldots,a_{k})\ \mbox{ if and only if }\ \mathbb{Z}_{m}\models\Phi(a_{1},\ldots,a_{k}).$$

Conversely, for every formula $\phi(x_1, \ldots, x_k)$  of
${\cal R}ing(0,+,*,<)$, there exists a formula 
$\Phi(x_1, \ldots, x_k)$ of $\mbox{FO}(\le,\oplus, \otimes)$  such that for every
finite structure $\mathbb{Z}_{m}$ and integers $a_{1},\ldots,a_{k}<m$,
$$\mathbb{Z}_{m}\models\phi(a_{1},\ldots,a_{k})\ \mbox{ if and only if }\ \mathcal{A}_{m}\models
\Phi(a_{1},\ldots,a_{k}). \qquad \Box$$ 
\end{theorem}
 The result also applies to the respective extensions of the logics with 
 modular quantifiers and the majority quantifier.
 
 \begin{remark}\label{def4Lring}
 Definability (or expressibility) in the logic of finite rings is given in terms of the finite residue structures $\mathbb{Z}_{m}$. That is, whenever we say that 
 a property of integers $P(x)$ is definable in ${\cal R}ing(0,+,*,<)+MOD+M$, or any fragment 
 $\cal L$ thereof, we mean that there exists a sentence $\varphi$ of 
 ${\cal L}$  such that for every natural $m$, \\
\centerline{$ P(m) \mbox{ holds in }\mathbb{Z} \iff \mathbb{Z}_{m}\models \varphi .$}
 
For a given  circuit class $\cal C$, we say that it is 
definable in 
the ring logic $\cal L$ if every property $P(x)$ decidable in $\cal C$ is definable in $\cal L$ and, for every sentence $\varphi$ in $\cal L$, 
the set of natural numbers $m$ such that $ \mathbb{Z}_{m}\models \varphi$
is decidable in $\cal C$. \qed
 \end{remark}
 
 As a consequence of the logical equivalence in Theorem \ref{ring=fo:order}, 
 any separation
result proved for fragments of ${\cal R}ing(0,+,*,<)+MOD+M$ can be translated into a corresponding separation
result in fragments of $\mbox{FO}(\le,\oplus, \otimes)+MOD+M$, with the respective implications to
circuit complexity. 
More specifically, from Theorem \ref{ring=fo:order} and 
Theorem \ref{bisthm}
we have the following definability  of uniform circuit classes in ring logics.

\begin{theorem}\label{RINGcaptCC} 
1) DLOGTIME-uniform $AC^{0}$ is 
 definable by ${\cal R}ing(0,+,*,<)$.\\
2) DLOGTIME-uniform $ACC(q)$ is 
definable by ${\cal R}ing(0,+,*,<)+MOD(q)$, for every natural $q$.\\
3) DLOGTIME-uniform $ACC$ is 
definable by ${\cal R}ing(0,+,*,<)+MOD$.\\
4) DLOGTIME-uniform $TC^{0}$ is 
definable by ${\cal R}ing(0,+,*,<)+MOD+M$. \qed
\end{theorem}


Our purpose is to work in the theory of ${\cal R}ing(0,+,*,<)$ to exploit many
algebraic properties and results of classes of residue rings, and in particular of finite fields.

\section{The prime spectrum of a sentence and systems of polynomial congruences}\label{S:RingSp}
 
 \begin{definition}\label{def-Sp}
 The prime spectrum of a sentence $\sigma$ of ${\cal R}ing(0,+,*,<)+MOD+M$
 is defined as the  set of primes
 $Sp(\sigma) = \{p \in \mathbb{P} : \mathbb{Z}_p \models \sigma \}$.
  \end{definition}
  The set $Sp(\sigma)$ was introduced by James Ax 
  in connection with his proof
  of decidability of the theory of finite fields \cite{ax}. In particular, 
  Ax proved the following:
  \begin{theorem}[\cite{ax}]\label{thmax1}
The spectrum $Sp(\sigma)$ of any sentence $\sigma$ of ${\cal R}ing(0,+,*)$ is, up to finitely many exceptions, a Boolean
combination of sets of the form $Sp(\exists t(f(t) = 0))$, where $f(t)\in \mathbb{Z}[t]$ is a polynomial with integer coefficients. \qed
\end{theorem}

Therefore to characterize the spectra of sentences of ${\cal R}ing(0,+,*)$ 
it is sufficient to 
analyze the spectra of sentences of the form $\exists x (f(x) = 0)$ for polynomials
$f \in \mathbb{Z}[x]$. Given a polynomial $f(x) \in \mathbb{Z}[x]$ we will indistinctly denote $Sp(f)$ or 
$Sp(\exists x(f(x) = 0))$ the spectrum of the 
sentence $\exists x(f(x) = 0)$.
A basic result of Schur  states that every non constant polynomial has an infinite number of prime divisors; that is, $Sp(f)$ is infinite 
for any $f \in \mathbb{Z}[x] \setminus \mathbb{Z}$ 
(see 
\cite[Thm. 1]{GB71} for an elementary proof of this fact). 
If $f$ is irreducible then  the same can be said about the complement of $Sp(f)$, 
namely, 
$$Sp(f)^c := Sp(\forall x(f(x) \not=0)) = \{p \in \mathbb{P} : \mathbb{Z}_p \not\models 
\exists x (f(x) = 0)\}$$
Thus,  we have the following properties of spectra of the form $Sp(f)$.
\begin{theorem}
$(1)$ For any $f \in \mathbb{Z}[x] \setminus \mathbb{Z}$, $Sp(f)$ is infinite.\\
$(2)$ If, additionally, $f$ is irreducible, then $Sp(f)^c$ is infinite. \qed
\end{theorem}
This theorem justifies to consider two spectra as equal
if they coincide in all but a finite number of primes, and we denote this by
$Sp(f) =^* Sp(g)$. 
 Moreover, we denote by 
$Sp(\sigma) \subseteq^* Sp(\theta)$ the fact that {\em almost all} 
of $Sp(\sigma)$ is contained in $Sp(\theta)$ (i.e. all but a finite number of primes).

The following result  from \cite{GB71}
 will be useful for obtaining further properties of spectra, by exploiting the 
 relation between irreducible polynomials and   algebraic finite extensions of 
 the rational field $\mathbb{Q}$.    
 (These extensions can be defined by adjoining to $\mathbb{Q}$ a root of a 
 polynomial irreducible over $\mathbb{Q}$. Such a root is called a 
 {\em primitive element} of the extension.)
 
\begin{theorem}[{\cite[Thm. 2]{GB71}}]\label{field2Sp}
Let $\mathbb{Q}$ be the rational field and $f(x)$ and $g(x)$ two non constant irreducible 
polynomials in $\mathbb{Q}[x]$. If $R$ and $S$ are algebraic extensions of 
$\mathbb{Q}$ such that $R\subseteq S$, $f$ has a root which is a primitive element of $R$ and $g$ has a root which is a primitive element of  $S$, 
then $Sp(g) \subseteq^* Sp(f)$. \qed
\end{theorem}

Using   Theorem \ref{field2Sp} we can prove that
the intersection of prime spectra of polynomial congruences contains a prime spectrum of some polynomial congruence, and hence it is also infinite. 
\begin{theorem}\label{meet4Sp} 
If $f_1$, $f_2$, \ldots, $f_k$ $\in  \mathbb{Z}[x] \setminus \mathbb{Z}$ irreducible polynomials, 
then there is $g(x) \in  \mathbb{Z}[x] \setminus \mathbb{Z}$ such that 
$$Sp(g) \subseteq^* Sp(f_1) \cap Sp(f_2) \cap \ldots \cap Sp(f_k)$$
\end{theorem}
\begin{proof}
We set $k=2$, the general case being similar. 
Let $\alpha_1$ and $\alpha_2$ be zeros of $f_1$ and $f_2$ respectively in $\overline{\mathbb{Q}}$ (the algebraic closure of the rationals). 
Consider the algebraic extensions
 $\mathbb{Q}(\alpha_1)$ and $ \mathbb{Q}(\alpha_2)$.   
 Take $\alpha \in \overline{\mathbb{Q}}$ 
 such that $\mathbb{Q}(\alpha_1, \alpha_2) \subseteq \mathbb{Q}(\alpha)$, and let $g(x) \in \mathbb{Z}[x]$ be an irreducible polynomial with 
 $g(\alpha) = 0$. Then $\mathbb{Q}(\alpha_1) \subseteq \mathbb{Q}(\alpha)$ and $\mathbb{Q}(\alpha_2) \subseteq \mathbb{Q}(\alpha)$, and by Theorem \ref{field2Sp}, $Sp(g) \subseteq^* Sp(f_1)$ and 
 $Sp(g) \subseteq^* Sp(f_2)$. The result now follows. \qed
 \end{proof}
 
\section{From systems of polynomial congruences to sets of 
congruent integers}\label{sec:5}

  \begin{example} From the Quadratic Reciprocity Law \cite[Ch. IV]{Nagel},   
  the prime spectrum 
  of the sentence $\exists x(x^{2}+1=0)$ is almost identical to the 
  set $\{p\in \mathbb{P}:p\equiv_{4}1\}$, i.e. 
  $$
Sp\left(\exists x(x^{2}+1=0)\right)=^{*}\{p\in \mathbb{P}:p\equiv_{4}1\}.
$$
  \end{example}
This number theoretical characterization of the solution set of certain 
Diophantine equations is part of a large body of knowledge within algebraic
number theory, 
  from where we obtain several tools to classify 
the spectra of ring formulae. 
As a further illus\-tration of this point consider
the following  basic relation between prime spectra and sets of congruent integers,
given by  the
 prime divisors of the $n$th cyclotomic polynomial $F_n(x)$, and shown in \cite[Thm. 94, p. 164]{Nagel}. The polynomial $F_n(x)$ 
 is the monic polynomial whose roots are the primitive $n$th roots of unity.
 
\begin{theorem}\label{cyclo-thm}
For $n > 1$, let $F_n(x)$ be the $n$-th cyclotomic polynomial. Then
$Sp(F_n) =^* \{p \in \mathbb{P}: p \equiv_n 1\}$. \qed
\end{theorem}
This result implies that, for each $n>1$, the set  $\{p \in \mathbb{P}: p \equiv_n 1\}$ is definable in the theory of finite residue class rings by the elementary 
sentence $\exists x (F_n(x) = 0)$. This does not goes through with all congruence
of the form $p \equiv_n r$, for $r>1$, as we show below.
 
\begin{corollary}\label{notcong}
For integers $n > 2$ and $1 < r < n$, the subset of prime numbers
$\{p \in \mathbb{P} : p \equiv_n r\}$ is not the prime spectrum of an irreducible polynomial over a field. 

Furthermore, for two integers $n, m > 1$ and $1< r < n$, $1< t < m$, the union \\
$\{p\in \mathbb{P} : p \equiv_n r\} \cup \{p\in \mathbb{P} : p \equiv_m t\}$ can not be the prime spectrum of a  polynomial over a field. 
\end{corollary}
\begin{proof}
If for some irreducible polynomial $g(x)$ over a field $K$, we have 
$$\{p\in \mathbb{P} : p \equiv_n r\} =^* Sp(g) ,$$ then considering the $n$th cyclotomic polynomial $F_n(x)$ over $K$, we have
$$\{p\in \mathbb{P} : p \equiv_n r\} \cap \{p\in \mathbb{P} : p \equiv_n 1\} =^* Sp(g) \cap Sp(F_n)$$
But the intersection on the left hand side of the above equality is empty, whilst the intersection on the right hand side is infinite by Theorem
\ref{meet4Sp}, and we have a contradiction. 
To prove the second statement repeat the previous argument with 
the $n$th and the $m$th cyclotomic polynomials together. \qed
\end{proof}

The next example shows that the unrepresentability of a   spectrum by congruences, stated in Corollary \ref{notcong}, does not holds in general for Boolean combinations of spectra. 
\begin{example}\label{exa:1}
Consider the polynomials $f(x) = x^2 + 1$ and $g(x) = x^2 -2$. Using the Quadratic Reciprocity Law we can see that
$Sp(f) = \{p\in \mathbb{P}  : p = 2 \vee p \equiv_8 1  \vee p \equiv_8 5\}$ and
$Sp(g) = \{p\in \mathbb{P}  : p = 2 \vee p \equiv_8 1  \vee p \equiv_8 7\}$. 
Now, $Sp(g)^c = \{p\in \mathbb{P} : p \equiv_8  2 \vee p\equiv_8 3 \vee p\equiv_8 4 \vee p\equiv_8 5 \vee p\equiv_8 6\}$, and
$$Sp(g)^c \cap Sp(f) = \{p\in \mathbb{P}  : p \equiv_8 5\}$$
\end{example}

We now jump to a more definitive result in this line of research on sets of primes
determined by polynomial congruences.  
Lagarias in \cite{Lag83} 
considered the sets $\Sigma(S)$ of prime divisors of 
 systems $S$ of polynomial congruences,  and the  
 Boolean algebra ${\cal B}$ generated by these sets.
 The Boolean algebra $\cal B$ corresponds to the collection of spectra of sentences in ${\cal R}ing(0,$ $+,$ $*)$ which, by Theorem \ref{thmax1}, collapses to the Boolean algebra generated by sets $\Sigma(S)$ where the 
polynomials in $S$ are restricted to have at most $1$ variable.
Then, Lagarias gave the following  characterization of the sets of integer congruences   
$\{p\in \mathbb{P}  : p\equiv_{d}a\}$,
 for given positive integers $d$ and $a$,
that are in ${\cal B}$. 

\begin{theorem}[{\cite[Thm. 1.4]{Lag83}}]\label{thm:lagar}
For any pair of integers $a$ and $d$, the set 
$\{p \in \mathbb{P} : p\equiv_{d}a\}$ is in the Boolean algebra 
$\cal B$ if and only if 
 $a$ is of order $1$ or $2$ in $\mathbb{Z}_d$ (i.e. $a\equiv_d 1$ or $a^{2}\equiv_d 1$),
or $(a,d)>1$. \qed
\end{theorem}
Rephrasing this theorem in terms of spectra of sentences we obtain
the following result.
\begin{theorem}\label{Lagarias}
For any pair of positive integers $a$ and $d$, with $1< a < d$, the set 
$\{p \in \mathbb{P}  :  p\equiv_{d}a\}$ is the spectrum of 
a sentence  in  ${\cal R}ing(0,+,*)$
 if and only if 
 $a^{2}\equiv_d 1$   
 or $(a,d)>1$. \qed
\end{theorem}
This theorem allow us to show some undefinability results.   
Any set of the form
$\{p \in \mathbb{P} : p\equiv_{d}a\}$, for $1<a<d$, $(a,d)=1$ and 
$a^2 \not\equiv_d 1$, is not the spectrum   of some sentence of 
${\cal R}ing(0,+,*)$.
For example,  $\{p : p\equiv_5 2\}$ is not in the spectra of ring sentences.
On the other hand observe that $5^2 = 25 \equiv_8 1$, and according to the theorem the set $\{p : p\equiv_8 5\}$ is the spectrum of some sentence of 
${\cal R}ing(0,+,*)$, a fact we already knew from explicit calculations in  Example
\ref{exa:1}. 
We show in the next subsections
that these sets of primes are definable in the extension of 
${\cal R}ing(0,+,*)$ with modular quantifiers, and in the ordered extension; hence separating ${\cal R}ing(0,+,*)$
  from  these logical extensions.

\subsection{The spectra of sentences in ${\cal R}ing(0,+,*)+MOD$}

\begin{remark} 
In ${\cal R}ing(0,+,*)+MOD(d)$ we have that $\forall a<d$, 
$$
Sp\left(\exists^{a,d}(x=x)\right)=^{*}\{p\in \mathbb{P}:p\equiv_{d}a\}.
$$
\end{remark}

Therefore, by Theorem \ref{Lagarias}, if we can find for every $d$ an $1< a<d$ that is relatively prime to $d$, and such that $a^{2}\not \equiv_{d}1$, then we have a set of primes definable
 in ${\cal R}ing(0,+,*)+MOD(d)$ 
that is not definable in ${\cal R}ing(0,+,*)$.
The question is:  
{\it for which natural numbers $d$ there exists $1< a<d$, relatively prime to $d$ and such that $a^{2}\not\equiv_{d}1$?}
To answer this question we look first at the prime numbers. Fix $p\in \mathbb{P}$. Note first that if there exists $a<p$ with $a^{2}\not\equiv_{p} 1$ then for every $\alpha$, $a^{2}\not \equiv_{p^{\alpha}}1$. Note now that for every prime $p>3$ we have that $2^{2}=4\not\equiv_{p}1$.
Hence, for any prime $p>3$ and any $\alpha$ we have that
 $(2,p)=1$ and $2^{2}\not \equiv_{p^{\alpha}}1$.

Consider now an arbitrary integer $d$ and its prime decomposition: $d=p_{1}^{\alpha_{1}}p_{2}^{\alpha_{2}}\ldots p_{n}^{\alpha_{n}}$. 
If one of the $p_{i}$ is greater than 3 then 
$(2,p_{i})=1$ and $2^{2}\not \equiv_{p_{i}^{\alpha_i}}1$. 
We know that 
$$
\mathbb{Z}_d \cong\mathbb{Z}_{p_{1}^{\alpha_{1}}}
\times \mathbb{Z}_{p_{2}^{\alpha_{2}}} \times \ldots 
\times \mathbb{Z}_{p_{n}^{\alpha_{n}}}.
$$
Then note that the element $(1,\ldots,2,\ldots,1)$, with $2$ in the $i$-th coordinate and $1$ everywhere else,
is relatively prime to $d$ (the only elements that are not relatively prime to $d$ are the ones of the form $(a_{1},a_{2},\ldots,a_{n})$ where for some $i$, $a_{i}=0$ or $a_{i}=p_{i}^{\beta}$ with $1<\beta<\alpha_{i}$.  
Also note that $(1,\ldots,2,\ldots, 1)^{2}\not \equiv_{d}(1,\ldots,1,\ldots,1)$.

 Looking now at powers of $3$ and $2$, note that $3^{2}\equiv_{2^{4}}9\not \equiv_{2^{4}} 1$, and that $2^{2}\equiv_{3^{2}}4\not \equiv_{3^{2}}1$. Hence, for any $d$ such that there is a prime $>3$ that divides $d$, or $3^{2}$ or $2^{4}$ divides $d$, we have that there exists $a<d$ such that $(a,d)=1$ 
 and $a^{2}\not \equiv_{d}1$. Hence, by Theorem  \ref{Lagarias} for such a $d$, 
 $\{p\in \mathbb{P}: p\equiv_{d}a\}$ is not expressible in 
 ${\cal R}ing(0,+,*)$. 

We  summarize these remarks in the following propositions.
\begin{proposition}
For every natural number $d\neq2^{\alpha}3^{\beta}$, $0\leq \alpha\leq 3$, $0\leq \beta\leq 1$, there exists $a<d$ with $(a,d)=1$ 
and $a^{2}\not\equiv_{d}1$. \qed
\end{proposition} 

 \begin{proposition}\label{enero}
 For every natural number $d \neq 2, 3, 4, 6, 8,12, 24\ $ there exists $a<d$ such that there is no sentence $\theta \in {\cal R}ing(0,+,*)$ equivalent to $\exists^{a,d}(x=x)$. 
Hence, in terms of expressive power, for every $d\neq2,3,4,6,8,12,24$,\\  \centerline{${\cal R}ing(0,+,*) \subsetneq {\cal R}ing(0,+,*)+MOD(d)$. \qquad $\Box$}
 \end{proposition} 

The problem with $d=2, 3, 4, 6, 8, 12, 24$ is that for each one of these $d$, $\forall a \in \mathbb{Z}_{d}$, either $a$ and $d$ are not relatively prime, or $a^{2}\equiv_{d}1$.
Hence for such integers we can not use the canonical counterexample above to separate ${\cal R}ing(0,+,*)$ from $ {\cal R}ing(0,+,*)+MOD(d)$.
However, we have obtained the desired inexpressibility for these integers 
(except $d=2$), 
through  direct combinatorial arguments.  

 For each one of the integer values of $d$ listed above, the key idea is to define in ${\cal R}ing(0,+,*)+MOD(d)$ a set of the form  $\{p: p\equiv_{nd}c\}$,  
such that $(c,nd) = 1$ and
$c^{2}\not \equiv_{nd}1$, for some integers $n$ and $c <nd$. Then, from 
Theorem \ref{Lagarias} we can  conclude  that this set is
 not expressible in ${\cal R}ing(0,+,*)$.

We are going to need some facts about power residues, and for the necessary
background we refer the reader to \cite[Ch. III, \S 34]{Nagel}.
First, we recall that for an integer $m\not= 0$ and integer $b$ prime to $m$, 
if $n \ge 2$ is a natural number such that $x^n \equiv_m b$ is solvable, then 
one says that $b$ is a {\em $n$-th power residue modulo $m$}.

The following theorem is an immediate consequence of Theorem 71 in \cite{Nagel}.
\begin{theorem}\label{thm:powresid}
Let $n \ge 2$ be a natural number and $p$ an odd prime such that $p\equiv_n 1$. 
Then there are $\frac{p-1}{n}$ $n$-th power residues incongruent modulo $p$
(i.e. the number of non-zero $n$-th powers in $\mathbb{Z}_p$ is $(p-1)/n$). \qed
\end{theorem}

We have now  a result that allows us to obtain expressibility of 
some sets of the form $\{p\in \mathbb{P}: p\equiv_{nd}r\}$ 
in ${\cal R}ing(0,+,*)+MOD(d)$.

\begin{theorem}\label{ladilla1} For every natural numbers $n,d>1$, for every 
$0\leq r<d$, there exists a sentence $\theta_{n,r}$ in ${\cal R}ing(0,+,*)+MOD(d)$ such that 
\[
Sp\left(\theta\right)=^* \{p \in \mathbb{P}: p \equiv_{nd} rn+1\}.
\]
\end{theorem}
\begin{proof} Fix $n,d>1$ and $0\leq r<d$. Using theorems \ref{cyclo-thm} and \ref{thm:powresid}, we have that for almost all primes $p$, 
\begin{eqnarray*}
p \equiv_{nd} rn+1 &\mbox{ iff }& 
  p\equiv_{n}1 \mbox{ and } \frac{p-1}{n}\equiv_{d}r   \\
&\mbox{ iff }&   \mathbb{Z}_{p}\models \exists x \left(F_n(x)=0\right) \wedge\exists^{r,d}y\exists z(z^{n}=y) 
\end{eqnarray*}
where $F_n(x)$ is the $n$-th cyclotomic polynomial. 
\qed
\end{proof}

We  use Theorem \ref{ladilla1} to obtain the desired inexpressibility results 
for $d= 3,$ $4,$ $6,$ $8,$ $12,$ $24$:
\begin{itemize}
\item For $d=3$, use $n=3$ and $r=1$. Then Theorem \ref{ladilla1} guarantees that $\{p\in \mathbb{P}:p\equiv_{9}4\}$ is expressible in 
${\cal R}ing(0,+,*)+MOD(3)$. Furthermore,  $(4,9)=1$ and $4^{2}\not \equiv_{9}1$
(so this set is not expressible in ${\cal R}ing(0,+,*)$).
\item For $d=4$, use $n=4$ and $r=1$. Theorem \ref{ladilla1} guarantees that 
$\{p\in \mathbb{P}:p\equiv_{16}5\}$ is expressible in ${\cal R}ing(0,+,*)+MOD(4)$. 
Furthermore, $(5,16) = 1$ and $5^{2}\not \equiv_{16}1$.
\item For $d=6$, use $n=3$ and $r=4$; for $d=12$, use $n=3$ and $r=2$;  and for $d=24$, use $n=3$ and $r=2$. 
\end{itemize}

This completes the proof for all the integer values of $d>2$. Thus, we can generalize Proposition \ref{enero} to obtain the following separation theorem:

\begin{theorem}
For all integers $d >2$, ${\cal R}ing(0,+,*) \subsetneq {\cal R}ing(0,+,*)+MOD(d)$. \qed
\end{theorem}

The remaining case, $d=2$, cannot be solved with the above ideas, and its solution is left as an open problem.

\subsection{The spectra of sentences in ${\cal R}ing(0,+,*,<)$ }

We now apply Theorem \ref{Lagarias} to show that the logic 
${\cal R}ing(0,+,*,<)$ 
strictly contains ${\cal R}ing(0,+,*)$. 
As in the previous section, it is enough to  show that  any set of the form 
$\{p \in \mathbb{P} : p\equiv_{d}a\}$  is  the spectrum of 
a sentence in ${\cal R}ing(0,+,*,<)$, although the
argument is more involved.

We review first some necessary modular arithmetic, and then give an example to illustrate how we are going to characterize
 the expression ``$p\equiv_d a$"
by a number theoretic fact involving order  which will   be easier to
define by a sentence of ${\cal R}ing(0,+,*,<)$.

\begin{remark}
Recall that for given (non negative) integers $a$ and $d$, 
if $(a,d) = 1$ then the inverse modulo $d$ of $a$ exists, and is the unique integer $a^{-1} < d$ such that $a^{-1}\cdot a \equiv_d 1$. 
By Euler's Theorem $a^{-1} \equiv_d a^{(\varphi(d)-1)}$, where $\varphi(n)$ is the number of positive integers less than $n$ and coprime  with $n$.
In particular, for any prime $p$, $\varphi(p) = p-1$, and hence, 
 the inverse modulo $p$ of an integer $a$ (which is not a multiple of $p$) is 
 $a^{p-2}$.
\end{remark}

\begin{definition}\label{def-frac}
Given a  prime field $\mathbb{Z}_{p}$, and two positive integers $a$ and $d$, 
with $d$ not a multiple of $p$, 
 the fraction $a/d$ in $\mathbb{Z}_p$ represents the unique integer $0<r<p$ such that $r\cdot d\equiv_{p}a$. In fact, by previous remark,  $r \equiv_p a\cdot d^{-1} \equiv_p a\cdot d^{(p-2)}$
\end{definition}

\begin{example}
Let $p=5$, $a=1$ and $d = 4$. Then $p \equiv_d a$. We are going to show that 
this congruence determines in $\mathbb{Z}_5$ an ordering of the fractions 
$1/4$, $2/4$, and $3/4$, and that this ordering implies the congruence.
Using Euler's Theorem, compute in $\mathbb{Z}_5$ the fractions:
\[\frac{1}{4} \equiv_5 1\cdot 4^{5-2} \equiv_5 4 , \quad
\frac{2}{4} \equiv_5 2\cdot 4^{3} \equiv_5 3 ,
\]
\[ \mbox{and }\quad \frac{3}{4} \equiv_5 3\cdot 4^{3} \equiv_5 2 .
\]
Thus, in $\mathbb{Z}_5$, $3/4 < 2/4 < 1/4$. Note that $(d-a)/d = 3/4$.

Rewrite the fractions $i/4$ in the form $(k5 + i)/4$, $i = 1,2,3$. 
To find the appropriate $k$, use that $5\equiv_4 1$; then, $k \equiv_4 1^{-1}(4 -i)$. Thus, in $\mathbb{Z}_5$, it holds that
 \[
\frac{1}{4} \equiv_5 4 \equiv_5 \frac{3\cdot 5 +1}{4} , \quad
\frac{2}{4} \equiv_5 3 \equiv_5 \frac{2\cdot 5 +2}{4}, 
\]
\[ \mbox{and }\quad \frac{3}{4} \equiv_5 2 \equiv_5 \frac{1\cdot 5 +3}{4} .
\]
In this form, we see that the ordering of the fractions $1/4$, $2/4$, and $3/4$, is in correspondence with the value of the coefficient $k$ in $(k5 + i)/4$ (and this holds because $p > d > a$). 
On the other hand, the smallest of the fractions, namely $3/4$, is equivalent (mod 5)
to $(5+3)/4$ (the coefficient of $5$ is $k = 1$), and this should be an integer. Note that $5+3 = p + (d-a)$, so $p + (d-a) \equiv_d 0$, or equivalently $p \equiv_d a$. 
Thus, we see that  the congruence $5 \equiv_4 1$ is characterized by the 
fact of  $3/4$ being the smallest in the order of the fractions $1/4$, $2/4$, and $3/4$ in $\mathbb{Z}_5$. \qed
\end{example}

Now, let us formalize the intuition  presented in the previous example.
The key tool is the following lemma.

\begin{lemma}\label{lem-key}
Let $p$ be a prime number and $d$ an  integer, such that $0< d < p$. Then:
\begin{itemize}
\item[$(i)$] For every $i = 1, 2, \ldots, d-1$, there exists a unique positive integer $k_i < d$,
such that\quad $\displaystyle \frac{i}{d} \equiv_p \frac{k_ip + i}{d}$.
\item[$(ii)$] The smallest of the
fractions $\displaystyle 1/d, 2/d, \ldots, (d-1)/d$ in $\mathbb{Z}_p$ is $c/d$, 
where $c$ satisfies $0 < c < d$ and $p+c\equiv_d 0$.
\end{itemize}
 
\end{lemma}
\begin{proof}
$(i)$: Fix $i < d$. By Definition \ref{def-frac}, $i/d$ is an integer $<p$ such that  $(i/d)\cdot d \equiv_p i$. 
Let $k_i < d$ such that $k_ip+i \equiv_d 0$ (always exist since $(d,p) = 1$). 
Then $\displaystyle \left(\frac{k_ip + i}{d}\right)\cdot d \equiv_p i$.

\noindent
$(ii)$: From $(i)$ we know $\displaystyle \frac{i}{d} \equiv_p \frac{k_ip + i}{d}$, for every $i < d$. The smallest of such numbers is when $k_i=1$, which is obtained 
when $p+i \equiv_d 0$. \qed
\end{proof}

\begin{theorem}\label{thm-cong2order}
Fix integers $a$ and $d$, with $0<a<d$.   
For every prime $p$, with $p > d$, we have that
\begin{quote}
$p \equiv_d a$ if,  and only if, in $\mathbb{Z}_p$, $(d-a)/d$ is the smallest fraction of the 
set $\{1/d, 2/d, \ldots, (d-1)/d\}$.
\end{quote}
\end{theorem}
\begin{proof}
If $p \equiv_d a$ then $\frac{p+ d-a}{d} \equiv_d 0$. On the other hand, $d-a < p$
and so $\frac{p+ d-a}{d} \equiv_p \frac{d-a}{d}$, and by 
Lemma \ref{lem-key} $(ii)$ this is the smallest fraction from 
  the set 
$\{1/d, 2/d, \ldots, (d-1)/d\}$ in $\mathbb{Z}_p$.

Conversely, assume  $(d-a)/d$ is the smallest fraction. 
By Lemma \ref{lem-key} and the hypothesis, 
$\displaystyle \frac{d-a}{d}  \equiv_p \frac{p + (d-a)}{d}$.
Therefore   $ \frac{p + (d-a)}{d}$ is an integer, and hence, 
$p + (d-a) \equiv_d 0$, which implies that $p \equiv_d a$.
\end{proof}

The fraction $a/d$ can be  defined by the following formula of 
${\cal R}ing(0,+,*,<)$:
\begin{eqnarray*}
\exists x\, (x*d=a) .
\end{eqnarray*}
Abusing notation we will denote by 
$(a/d)<(b/d)$
the sentence  
\begin{eqnarray*}
\exists x \exists y ( x*d=a \wedge y*d=b \wedge x<y)
\end{eqnarray*}

\begin{theorem}\label{thm-main}
For every positive integers $a$ and $d$, with $a<d$, there exists a sentence
$\theta$ of ${\cal R}ing(0,+,*,<)$, such that for every prime $p > d$, 
$$p\equiv_d a\ \mbox{ if and only if }\ \mathbb{Z}_p \models \theta$$
In other words, $Sp(\theta) = \{p \in \mathbb{P} : \ p \equiv_d a\}$.
\end{theorem}
\begin{proof}
By Theorem \ref{thm-cong2order} we just have to write a sentence $\theta$ that expresses the fact that ``$a < d$ and $(d-a)/d$ is the smallest fraction of the 
set $\{1/d, 2/d, \ldots, (d-1)/d\}$". Here it is:
$$(a < d) \wedge \bigwedge_{0 < r < d \atop r \not= d-a} (\ (d-a)/d < r/d\ )
$$
\end{proof}

The previous result, together with Theorem \ref{Lagarias}  establishes the following fundamental difference between the logics ${\cal R}ing(0,+,*)$ and  ${\cal R}ing(0,+,*,<)$:

\begin{theorem}\label{thm:ringNOringorder}
 The logic ${\cal R}ing(0,+,*)$ is  weaker than the logic ${\cal R}ing(0,+,*,<)$. \qed
\end{theorem}

Having discerned the expressive power of  ${\cal R}ing(0,+,*)$ with respect to its
extension with modular quantifiers, on the one hand, and order on the other hand, the next step is to explore the expressive power of  ${\cal R}ing(0,+,*,<)$ (the logic of rings with order) with respect to its extensions with modular  and majority  
quantifiers. In this case the manipulations of specific sets of integer congruences will be of no use, since in 
the presence of order or with the generalized quantifiers we can express these sets of congruences beyond those restricted by Theorem \ref{thm:lagar}. 
Hence we need to introduce another tool for separating these logics, and this will be based on the analytical notion of density.

\section{The  density of the prime spectrum of a sentence}\label{sec:6}
A way of discerning infinite sets of primes is to compare 
their relative sizes. For that matter a measure of density of subsets of natural numbers will come in hand. There are various notions of density, 
but in this work we deal only with the {\em natural} and the {\em exponential} densities. (See the survey \cite{Greko05}. However, we note that our definition of density differs  from \cite{Greko05} in that ours are relative to the set 
of all primes, as opposed to all natural numbers.)

 \begin{definition} For a positive real $x$, 
 $\pi(x)=|\{q \in \mathbb{P} : q\le x\}| = |\mathbb{P} \cap [1,x]|$ is the counting function of primes below $x$. 
Similarly, for a given subset $S \subset \mathbb{P}$ define
$\pi_S(x)=|\{q \in S : q\le x\}| = |S \cap [1,x]|$;
 and for a given sentence $\psi$ of the logic of rings,
 ${\cal R}ing(0,+,*,<) + MOD+M$, we define
$\pi_{\psi}(x)= |\{q \in \mathbb{P} : q\le x \wedge 
\mathbb{Z}_q\models \psi\}|.$
\end{definition}

The Prime Number Theorem (PNT) states that for all $x>0$, $\pi(x)$ is asymptotic to 
$x/\log x$; that is
\begin{equation*}
\lim_{x \to +\infty} \frac{\pi(x)\log x}{x} = 1 \qquad \mbox{ (PNT)}
\end{equation*}
This is equivalent to saying that for all $\epsilon >0$, there is $N>0$, such that for all 
$x > N$, $\displaystyle (1-\epsilon) \frac{x}{\log x} < \pi(x) < (1+\epsilon) \frac{x}{\log x}$. In what follows we will fix $\epsilon= 1/2$ and work with the following bounds for $\pi(x)$, which hold for almost all $x$:
\begin{equation}\label{PNTbnd}
\frac{1}{2}\frac{x}{\log x} < \pi(x) < \frac{3}{2}\frac{x}{\log x}
\end{equation}

We work with a family of  densities   
given by the following definition.
\begin{definition}\label{hdensity}
Let $h$ be a real positive, continuous, unbounded and increasing function defined on $(0,+\infty)$.
For a given $S \subset \mathbb{P}$, the lower $h$-density of $S$ is defined by 
\begin{equation*}
\underline{\delta}_h(S) = \liminf_{n\to \infty}\frac{h(\pi_S(n))}{h(\pi(n))}
\end{equation*}
and its upper h-density by
\begin{equation*}
\overline{\delta}_h(S) = \limsup_{n\to \infty}\frac{h(\pi_S(n))}{h(\pi(n))}
\end{equation*}
If these limits are equal, i.e., $\underline{\delta}_h(S) = \overline{\delta}_h(S)$, we say that the set $S$ has $h$-density, and its value is the limit 
$\displaystyle\delta_h(S) = \lim_{n\to \infty}\frac{h(\pi_S(n))}{h(\pi(n))}$.
\end{definition}
The basic properties of  an $h$-density  are the following\footnote{For practical reasons we omit dealing with the empty set; in general, one sets 
$\delta_h(\emptyset) = 0$.}. 
\begin{itemize} 
\item If $S$ is finite then   $\delta_h(S) = 0$.
 \item $\delta_h(\mathbb{P}) = 1$.
\item If $S$ and $T$ are two sets of primes such that 
$S \subseteq T$ and both sets have $h$-density, then $\delta_h(S) \leq \delta_h(T)$ (monotonicity); 
and if $S=^* T$ then
$\delta_h(S) = \delta_h(T)$.
\end{itemize}

Two cases of $h$-densities are of particular interest to us. When $h$ is the identity function we have the {\em natural } density, denoted as $\delta(S)$, when the limit exists, i.e., 
\[ \delta(S) = \lim_{n\to \infty}\frac{\pi_S(n)}{\pi(n)}
\]
The lower and upper natural densities, $\underline{\delta}(S)$ and 
$\overline{\delta}(S)$, are defined accordingly taking $\limsup$ and $\liminf$.
The other case of interest is when $h = \log$, then we have the {\em exponential} density 
denoted $\varepsilon(S)$ when the limit exists, i.e.,
\[ \varepsilon(S) = \lim_{n\to \infty}\frac{\log(\pi_S(n))}{\log(\pi(n))}
\]
and we always  have the  lower and upper exponential  densities,
 $\underline{\varepsilon}(S)$ and 
$\overline{\varepsilon}(S)$, defined by taking  $\limsup$ and $\liminf$.
(The reason for the name {\em exponential}, given in \cite{Greko05} for the unrelativized version of $h$-density, is that the exponential density $\varepsilon$ acts as a magnifying  glass on subsets of the naturals with natural 
density zero.) 

The following observations are useful for making further calculations.

\begin{remark}\label{rem:61}
Recall the notation $f(x) \thicksim g(x)$ means $f(x)$ is asymptotic to 
$g(x)$, this means  that $\lim_{x\to +\infty} \frac{f(x)}{g(x)} = 1$.
Some useful properties of $\thicksim$ are:

1) If $f$, $g$, $h$, $k$ are all real value functions such that $f(x) \thicksim g(x)$
 and $h(x) \thicksim k(x)$ and $\lim_{x \to +\infty} \frac{f(x)}{h(x)}$ 
exists or is $+\infty$, then $\lim_{x \to +\infty} \frac{f(x)}{h(x)} = \lim_{x \to +\infty} \frac{g(x)}{k(x)}$.

2) For any two real value functions $f(x)$ and $g(x)$, with $\lim_{x\to +\infty} g(x) = +\infty$, if $f(x) \thicksim g(x)$ then 
$\log f(x) \thicksim \log g(x)$.
This can be seen from the following equalities: 
\[ \frac{\log\left(\frac{f(x)}{g(x)}\right)}{\log g(x)} = \frac{\log f(x) - \log g(x)}{\log g(x)}
= \frac{\log f(x)}{\log g(x)} - 1 \]
As $x \to +\infty$ the first term of these equalities goes to 0, and hence
$\lim_{x\to +\infty}\frac{\log f(x)}{\log g(x)} = 1$.

3) Using the previous observations and the PNT, we have that for $h$ the identity or the logarithm function it holds that
\begin{equation}\label{density2}
\delta_h(S) = \lim_{n \to \infty} \frac{h(\pi_S(n))}{h(\pi(n))} = 
\lim_{n \to \infty} \frac{h(\pi_S(n))}{h(n/\log n)} 
\end{equation}
\end{remark}

Observe that if a set $S \subset \mathbb{P}$ has natural density, then it has exponential 
density (i.e. $\delta(S) \Rightarrow \varepsilon(S)$). In fact, we have a stronger result:
\begin{theorem}\label{thm:delta2epsilon}
If $\delta(S)$ exists and is not zero then 
$\varepsilon(S) = 1$. 
\end{theorem}
\begin{proof}
Let $\alpha = \delta(S) = \lim_{n\to \infty}\frac{\pi_S(n)}{\pi(n)}$. Then, by the PNT, 
$\pi_S(n) \thicksim \alpha\frac{n}{\log n}$ and 
\begin{equation*}
\lim_{n\to \infty}\frac{\log(\pi_S(n))}{\log(\pi(n))} = \lim_{n\to \infty} \frac{\log \alpha + \log\frac{n}{\log n}}{\log\frac{n}{\log n}} = 1
\end{equation*}
using Remark (\ref{rem:61}).
\end{proof}


The Chebotarev's Density Theorem (cf. \cite[\S 5]{wyman}) implies that
every element of the Boolean algebra $\cal B$ 
has rational  natural density, and it is $0$ if and only if the set is finite. 
This together with Ax's result (Theorem \ref{thmax1}) gives: 
\begin{theorem}\label{thm:D4ring}
The spectrum of any sentence in ${\cal R}ing(0,+,*)$ has rational 
natural density, and this density is $0$ if and only if the spectrum is finite. \qed
\end{theorem}

We are going to prove that this is not the case for the spectra of sentences in ${\cal R}ing(0,+,*,<)$. We will see that  there exist sentences 
$\sigma\in {\cal R}ing(0,+,*,<)$ such that the natural density of $Sp(\sigma)$ is zero, but the cardinality of $Sp(\sigma)$ is 
infinite. 
This give us another way of showing that ${\cal R}ing(0,+,*)$ is properly
contained in ${\cal R}ing(0,+,*,<)$.

We begin by recalling 
an outstanding result by Friedlander and Iwaniec in 
\cite{FI97} (and further extended in \cite{FI98}) which shows that the polynomial $f(x,y) = x^2 + y^4$ has infinitely many prime values, but the sequence of its values is ``thin" in the sense that it contains fewer than $t^{\theta}$ integers up to $t$ for some $\theta < 1$. 
More specifically,
\begin{theorem}[\cite{FI97,FI98}]\label{thm-friedIwa}
There are infinitely many primes $p$ of the form $p = a^2 + b^4$, for integers $a$ and $b$, and the number of these primes $p < t$ is 
  $O(t^{3/4})$. \qed
  \end{theorem}
Using this result and Eq. (\ref{density2}), the natural density of the 
set of primes 
$$FI := \{p \in \mathbb{P} : p = a^2 + b^4,\ a,b \in \mathbb{Z} \}$$
 is 
$\displaystyle \lim_{t \to \infty} \frac{\log t}{t^{1/4}} = 0$.
By Theorem \ref{thm:D4ring}, such set $FI$ cannot be the spectrum of a sentence in 
 ${\cal R}ing(0,+,*)$. 
It remains to show that the set $FI$ is definable in ${\cal R}ing(0,+,*,<)$.
We will in fact show a stronger result. 
\begin{theorem}
Consider a polynomial in $\mathbb{Z}[x,y]$ of the form $f(x,y)=h(x)+g(y)$, with $n$ the degree of $h$  and $d$ the degree of $g$. 
Assume that $n,d \geq1$ and that the leading coefficients of $h(x)$ and $g(x)$ are positive. Then 
there is a sentence $\theta$ in ${\cal R}ing(0,+,*,<)$, such that 
for almost every $m$,   $\mathbb{Z}_m \models \theta$ if and only if in $\mathbb{Z}$  the following property holds: 
 \begin{center}
 ``There exists naturals $b,c<m$ such that $f(b+1,c)=m$''.
 \end{center}
\end{theorem}
\begin{proof}
The idea is that $h$ and $g$ will  be increasing from some threshold $M$ and onward, since their leading coefficients are positive. Then, for all $m > M$, 
such $b$ and $c$ will be characterized by a sentence in ${\cal R}ing(0,+,*,<)$ that says ``$f(b+1,c) = 0$ 
and $b+1 < m$ is the first element greater than $M$ such that $f(b+1,c) < f(b,c)$".

Note first that if $h(x)$ is a polynomial of degree $n\geq 1$ with positive leading coefficient, then for every $1\leq i\leq n$, the $i$th derivative of 
$h(x)$, $h^{(i)}(x)$, is a polynomial of degree $n-i$ with positive leading coefficient, hence eventually increasing. 
Hence there exists a natural $M$ such that $h$, $g$, and
 $h^{(1)},h^{(2)},\ldots, h^{(n)}$ are increasing and positive in the interval $[M,+\infty)$ and   for every $x>M$, 
\begin{equation}\label{eqh}
h(x)>h^{(1)}(x)+\frac{h^{(2)}(x)}{2!}+\ldots + \frac{ h^{(n)}(x)}{n!}
\end{equation}

Fix  a natural number $m>M+1$. We make the following claim.
\begin{quote}
{\it Claim: For every pair of integers $b,c \in (M,m-1)$, \\
if  $\mathbb{Z}\models h(b)+g(c)<m\leq h(b+1)+g(c)$ then in $\mathbb{Z}$: }
$$
h(b+1)+g(c)-m<h(b)+g(c)<m.
$$
\end{quote}
The proof of this claim is as follows.
By the Taylor polynomial expansion, 
$$
h(b+1)=h(b)+h^{(1)}(b)+\frac{h^{(2)}(b)}{2!}+\ldots +\frac{ h^{(n)}(b)}{n!} .
$$
Using that $h(b) < m$ and $b > M$, it follows  from Eq. (\ref{eqh})  that  
\begin{eqnarray*}
h(b+1)+ g(c)-m &\leq & h^{(1)}(b)+\frac{h^{(2)}(b)}{2!}+\ldots+
\frac{ h^{(n)}(b)}{n!} + g(c) \\
&<& h(b)+g(c)<m
\end{eqnarray*}
which is the the desired result. 

Now it follows that   there exists integers $b,c$, $M<b,c<m-1$ such that 
\begin{equation}\label{eqX}
\mathbb{Z}\models h(b)+g(c)<m\leq h(b+1)+g(c)
\end{equation} 
if and only if   we have:
\begin{itemize}
 \item $g(c)<m$, 
 \item for every $a$ with $M<a<b$, $h(a)+g(c)<h(a+1)+g(c)<m$, 
  \item $m\leq h(b+1)+g(c)<2m $,  and
\item $h(b+1)+g(c)-m<h(b)+g(c)<m$.
\end{itemize}
It follows that (\ref{eqX}) can be expressed by the following formula in 
${\cal R}ing(0,+,*,<)$:
\begin{eqnarray*}
\psi(b,c)&:=& M<b,c<m-1\wedge h(b+1)+g(c)<h(b)+g(c)\wedge \\
& & \forall a\left((M<a<b)\Rightarrow h(a)+g(c)<h(a+1)+g(c)\right)
\end{eqnarray*}
which says that $b+1<m$ is the first element bigger than $M$ for which 
$$f(b+1,c) = h(b+1)+g(c)< h(b)+g(c) = f(b,c)$$ in $\mathbb{Z}_{m}$. 
Putting together the previous observations we obtain that
$$
\mbox{ for every $m>M$, }\mathbb{Z}_{m}\models \exists b,c (\psi(b,c)\wedge f(b+1,c)=0)
$$
if and only if
$$
\mathbb{Z}\models \exists b,c(M< b,c<m-1\wedge f(b+1,c)=m).
$$
This completes the proof of the theorem.
\qed
\end{proof}

\begin{corollary}
${\cal R}ing(0,+,*)$ is properly contained in ${\cal R}ing(0,+,*,<)$. \qed
\end{corollary}

\section{A sufficient condition for the existence of sets without {$h$}-density}\label{sec:7}

Our overall goal is to classify spectra in terms of $h$-density; hence,  a first 
step is to elaborate some tools to determine when a set has no $h$-density.
We are going to canonically construct sets of primes without $h$-density from sequences of integers $\{s_{n}\}$. The idea is that certain combinatorial properties of these sequences will guaranteed that the associated set of primes do not have $h$-density.
The necessary combinatorial properties are defined around the notion of thinness of a sequence of numbers, 
which has been mentioned in the previous section in the intuitive form as employed by  
Friedlander and Iwaniec in their work \cite{FI97}. 

\begin{definition}\label{hthin}
Let $h$ be a real positive, continuous, unbounded and increasing function defined on $(0,+\infty)$.
An increasing  sequence of numbers $\{s_n\}_{n >0}$ is {\bf $h$-thin\/} if  
there exists some real $r > 3$ such that for almost all $n$, 
\[ r h(\pi(s_n)) < h(\pi(s_{n+1})) 
\]
When  $h$ is the identity function  we call the associated sequence thin.
\end{definition}
Essentially a sequence of numbers  is ``$h$-thin'' if the distance between the number of primes below consecutive elements  in the sequence, 
filtered through $h$, increases exponentially.

For the measures of density that we consider here, we focus on functions that are {\em eventually semi-additive}  in the following sense.
\begin{definition}\label{semiadd}
A function $h: {\cal D} \subseteq \mathbb{R} \to \mathbb{R}$ is {\bf eventually semi-additive\/}, if there exists $M >0$ such that for $x\ge y > M$, 
$h(x+y) \le h(x) + h(y)$.
\end{definition}

\begin{example}
The identity function is trivially eventually semi-additive, as well as any additive function. 
The function $h(x) = \log x$ is eventually semi-additive, since for all $x\ge y>2$ we have that $x + y \le xy$, and hence
\[\log(x+y) \le log(xy) = \log x + \log y\]
A similar argument applies to $\log\log x$, taking $M = e^2$.Hence this and other iterations of the logarithm function are  eventually semi-additive. \qed
\end{example}

\begin{remark}\label{conditionsemiadd}
We have the following properties for semi-additive functions $h$ defined on $(0,+\infty)$ that are real positive, continuous, unbounded and increasing.

\begin{itemize}
\item[$(i)$]  For all $x$ and $y$ such that 
$x > 2y>2M$, where $M$ is the bound in Definition \ref{semiadd}, we have $h(x-y) \ge h(x) - h(y)$.

Indeed, observe that $x-y > y>M$ and by semi-additive 
$h(x) = h(x-y+y) \le h(x-y) + h(y)$.

 \item[$(ii)$] If a sequence $\{s_{n}\}$ is $h$-thin then for every $\beta$, $2<\beta<3$,   for almost all natural numbers $n$ we have that:
 \[
\beta\pi\left(s_{n}\right)<\pi\left(s_{n+1}\right) .
\]
To see this, observe that by semi-additivity and the fact that $h$ is increasing we have that , for every $x>M$, $h(3x)\leq 3h(x)$, so for a  $\beta$, $2<\beta<3$, and for almost all natural numbers $n$, 
\[
h\left(\beta\pi\left(s_{n}\right)\right)\leq h\left(3\pi\left(s_{n}\right)\right)\leq 3h\left(\pi\left(s_{n}\right)\right)\leq h\left(\pi\left(s_{n+1}\right)\right).
\]
It follows that 
\[
\beta\pi\left(s_{n}\right)\leq \pi\left(s_{n+1}\right).
\]

\item[$(iii)$] If a sequence $\{s_{n}\}$ is $h$-thin then for almost all natural numbers $m<n$ we have that:
 \[
 h(\pi(s_{n})-\pi(s_m))\ge h(\pi(s_{n}))-h(\pi(s_m))
\]

To see this, note that if a sequence $\{s_{n}\}$ is $h$-thin, then for every $M$ there exists a bound $B_{M}$ such that if  $n>m>B_{M}$ we have by $(ii)$
 that $\pi(s_{n})>2\pi(s_{m}) >2M$.
Then apply $(i)$ above to get the desired result.
 \end{itemize}
 \end{remark}
 
 For every sequence of natural numbers $\{s_{n}\}$ we are going to construct an associated set of primes ${\cal H}(\{s_n\})$, which will be instrumental in showing 
 counterexamples to $h$-density for various classes of spectra.

\begin{definition}\label{alternating}
Fix an increasing sequence of natural numbers $\{s_{n}\}$. The {\bf alternating} set of primes associated with $\{s_{n}\}$, denoted by ${\cal H}(\{s_{n}\})$, is defined as:
\begin{eqnarray*}
{\cal H}(\{ s_{n} \}) &=& \{p \in \mathbb{P}: \forall n (s_{2n} < p < s_{2n+1})\} \\
&=& \mathbb{P} \cap ((s_2,s_3) \cup (s_4,s_5) \cup \ldots \cup (s_{2n}, s_{2n+1}) \cup 
\ldots )
\end{eqnarray*}
\end{definition}

 The following theorem gives conditions on the sequence $\{s_{n}\}$ that guarantee that the set ${\cal H}(\{ s_{n} \})$ has no $h$-density.
  
 \begin{theorem}\label{thm:Nodense}
Let $h: (0, +\infty) \to [0, +\infty)$ be   a real positive, continuous, unbounded and increasing function defined on $(0,+\infty)$. Assume additionally that $h$ is
eventually semi-additive.
Let $\{s_n\}_{n >0}$ be an increasing sequence of numbers that is $h$-thin
Then the set of alternating primes ${\cal H}(\{s_n\})$ 
has no $h$-density.
\end{theorem}
\begin{proof}
Put ${\cal H} = {\cal H}(\{s_n\})$. First note that
\[\pi_{\cal H}(s_{2n}) = \pi_{\cal H}(s_{2n-1}) + |\mathbb{P}\cap (s_{2n-1},s_{2n})| 
= \pi_{\cal H}(s_{2n-1}) \]
Then using that $h$ is increasing and the sequence is $h$-thin for a real $r>3$, we have for almost all $n$:
\begin{equation*}
\frac{h(\pi_{\cal H}(s_{2n}))}{h(\pi(s_{2n}))} = \frac{h(\pi_{\cal H}(s_{2n-1}))}{h(\pi(s_{2n}))}  \le \frac{h(\pi(s_{2n-1}))}{h(\pi(s_{2n}))} < \frac{1}{r}
\end{equation*}
Hence, $\displaystyle\liminf_{n\to +\infty} \delta_h({\cal H}) \le \frac{1}{r}$.
On the other hand, by definition of ${\cal H}(\{s_n\})$ 
\begin{equation*}
\pi_{\cal H}(s_{2n+1}) \ge \pi(s_{2n+1}) - \pi(s_{2n}) ,
\end{equation*}
hence, 
using that $h$ is increasing, Remark~\ref{conditionsemiadd} $(iii)$ and the fact that $\{s_{n}\}$ is $h$-thin, we have that for almost all $n$:
\begin{eqnarray*}
\frac{h(\pi_{\cal H}(s_{2n+1}))}{h(\pi(s_{2n+1}))} &\ge& 
\frac{h(\pi(s_{2n+1}) - \pi(s_{2n}))}{h(\pi(s_{2n+1}))}  \ge \\
& & \frac{h(\pi(s_{2n+1})) - h(\pi(s_{2n}))}{h(\pi(s_{2n+1}))}  = 1 - \frac{h(\pi(s_{2n}))}{h(\pi(s_{2n+1}))} > 1 - \frac{1}{r} 
\end{eqnarray*}
Hence, since $r>3$, $\displaystyle\limsup_{n\to +\infty} \delta_h({\cal H}) \ge 1-\frac{1}{r} > \frac{1}{r}$. 
Therefore,   $\underline{\delta}_h({\cal H}) \not= \overline{\delta}_h({\cal H})$ and the 
set $\cal H$ has no $h$-density. 
\end{proof}

The theorem above is useful  to prove that there exists a sentence  whose spectrum has no $h$-density: Find a sequence $\{s_{n}\}$ that is $h$-thin, and find a sentence $\theta$ such that $Sp(\theta)={\cal H}(\{s_n\})$.

We conclude this section with a criteria for a sequence $\{s_n\}$ to be thin. 
 
 \begin{lemma}\label{Laux}
 Let $\{s_n\}_{n>0}$ be a non decreasing sequence of numbers such that for some 
 $R > 18$, for all $n >N$, $R s_n < s_{n+1}$. 
 Then $r\pi(\{s_n\}) < \pi(\{s_{n+1}\})$, for some $r > 3$.
 \end{lemma}
 \begin{proof}
 Using Eq. (\ref{PNTbnd}) and that  $f(x) = x/\log x$ is strictly increasing,  we get
 \begin{equation*}
R \pi(s_n) < \frac{3}{2}\frac{R s_n}{\log\left(\frac{R}{R}s_n\right)} = 
\frac{3}{2\left(1-\frac{\log R}{\log(Rs_n)}\right)}\left(\frac{Rs_n}{\log (Rs_n)}\right)
< 3 \frac{s_{n+1}}{\log s_{n+1}} < 6\pi(s_{n+1})
\end{equation*}
Therefore, $r\pi(s_n) < \pi(s_{n+1})$ for $r = R/6 > 3$.
 \end{proof}

\section{On the density of spectra of ring formulae with order}\label{sec:8}

We use Theorem \ref{thm:Nodense} to show the existence of 
a sentence in ${\cal R}ing(0,+,*,<)$ whose spectrum has no natural density.

 \begin{theorem}\label{thm:nonat4R} There exists a sentence 
 $\psi \in {\cal R}ing(0,+,*,<)$ such that its spectrum $Sp(\psi)$ has no natural density.
\end{theorem}
\begin{proof}
Fix a prime $q \ge 19$ and consider the sequence $s_n = q^n$, for $n=2,3,\ldots$.
It is immediate that for $R$ such that $18 < R < 19$, $Rq^n < q^{n+1}$ and by 
Lemma \ref{Laux} the sequence is thin for $r > 3$. Then, by Theorem \ref{thm:Nodense}, the alternating set
\begin{equation}\label{setqn}
{\cal H}(\{q^n\})= \mathbb{P} \cap \left( (q^2, q^{3}) \cup (q^4, q^5) \cup \cdots \cup (q^{2k}, q^{2k+1}) \cup \cdots \right)
\end{equation}
has no natural density. 
 
We will define a sentence $\psi\in {\cal R}ing(0,+,*,<)$ so that for any structure 
$\mathbb{Z}_p$, 
$\mathbb{Z}_p \models \psi$ if and only if 
the size of the model is a prime $p$ such that $q^{2n}<p<q^{2n+1}$, for some natural $n$.
Then the spectrum of such sentence $\psi$ is  $Sp(\psi) = {\cal H}(\{q^n\})$.

From Theorem \ref{ring=fo:order} (see \cite{AACO13} for details) we know that $\otimes(x,y,z)$, true multiplication in $\mathbb{Z}_{m}$ (i.e. $\mathbb{Z}_{m}\models \otimes(x,y,z)
$ if and only if $\mathbb{Z}\models x\times y=z$)
 is expressible in ${\cal R}ing(0,+,*,<)$. 
 Hence, for a given prime $q$,
we can express in ${\cal R}ing(0,+,*,<)$
``{\it $z$ is a power of q}'' by saying that every divisor of $z$ different from 1 must be divisible by $q$ (this equivalence only holds if $q$ is prime). Here is  the formula:
\begin{equation}\label{fmlEXP}
EXP_{q}(z):=\forall x \left([\exists y (\otimes(x,y,z)\wedge x\neq 1]\Rightarrow [\exists w(\otimes(q,w,x)]\right)
\end{equation}
We can also have a formula $EXP_{q^2}(z)$ that says 
``{\em $z$ is a square of a power of $q$}":
\begin{equation}\label{fmlEXP2}
EXP_{q^{2}}(z):=\exists x\left(\otimes(x,x,z)\wedge EXP_{q}(x)\right)
\end{equation}
From these formulas we can express
the property ``{\it $z$ is the maximal power of $q$ in the structure}", 
by a formula $MAXEXP_{q}(z)$:
\begin{equation}\label{fmlMEXP}
MAXEXP_{q}(z):= EXP_{q}(z) \wedge \forall w(w>z\Rightarrow \neg EXP_{q}(w))
\end{equation}
and the property ``{\it $z$ is the maximal power of $q^{2}$ in the structure}" by the formula 
\begin{equation}\label{fmlMEXP2}
MAXEXP_{q^{2}}(z):= EXP_{q^{2}}(z) \wedge \forall w(w>z\Rightarrow \neg EXP_{q^{2}}(w))
\end{equation}
Finally, we need the sentence $PRIME$ which expresses that the size of the ring model is prime (it is enough to say that every element has a multiplicative inverse).
 We thus have, for a fixed prime $q\ge 19$, 
$$
\psi:=PRIME \wedge \exists z\left(MAXEXP_{q}(z)\wedge MAXEXP_{q^{2}}(z)\right) 
$$
the required sentence in ${\cal R}ing(0,+,*,<)$ such that 
$Sp(\psi) = {\cal H}(\{q^n\})$,  which has no natural density.
\end{proof}

 As an immediate corollary we have a third proof of the difference in expressive power of the ring logic with order and the ring logic without order.
 This result also attest to the strength of the logic ${\cal R}ing(0,+,*,<)$. 
 We had conjectured in \cite{AACO13} that  the spectrum of any sentence 
  in ${\cal R}ing(0,+,*,<)$ has  natural  density, and had shown there that the set
  ${\cal H}(\{q^n\})$ in (\ref{setqn}) is the spectrum of a sentence of 
${\cal R}ing(0,+,*,<) + MOD(2)$. Now, Theorem \ref{thm:nonat4R} refines that expressibility result and knocks down our conjecture. 
However, we believe in the following.

\begin{conjecture}\label{conj}
Every sentence in ${\cal R}ing(0,+,*,<)$ has exponential density. \qed
\end{conjecture}

In what follows we provide some evidence for this conjecture. Every sentence in the unordered fragment ${\cal R}ing(0,+,*)$ has exponential density. This is a consequence of the general fact that a set for which the natural density exists, it then has the exponential density, and in fact, this exponential density is always 1 (Theorem \ref{thm:delta2epsilon}).
Furthermore, all the properties that we have defined so far in ${\cal R}ing(0,+,*,<)$ have 
exponential density. 
In particular, 
 the set ${\cal H} := {\cal H}(\{q^n\})$ in Theorem \ref{thm:nonat4R} has  exponential density.
 Let us show this. 
 
 Let $\displaystyle s_n = \frac{\log\pi_{\cal H}(n)}{\log\pi(n)}$. It will be sufficient to show that $\liminf_{n \to +\infty} s_n = 1$.
 
First observe that, for every $n$, choosing $m$ such   that 
$q^{m}\leq n < q^{m+1}$, 
we obtain the following inequalities due to the definition of ${\cal H}(\{q^n\})$:
$$
\pi(q^{m-1})-\pi(q^{m-2})\leq \pi_{\cal H}(n) 
$$
Hence $\displaystyle \frac{ \pi_{\cal H}(n)}{ \pi(n)} \ge \frac{\pi(q^{m-1}) - \pi(q^{m-2})}{\pi(q^{m+1})}$. 
 Now, observe that, in general, if $1 < y < x$ then
 \[ \pi(x) - \pi(y) = \sum_{y<p\le x} 1 \ge \sum_{y<p\le x} \frac{\log p}{\log x} \ge 
 \frac{(x-y)\log y}{\log x} \]
 Therefore, for any $z >1$, 
 \begin{equation}\label{cota}
\frac{\log(\pi(x) -\pi(y))}{\log \pi(z)} \ge \frac{\log\left(\frac{x-y}{\log x}\right) + \log\log y}{\log \pi(z)} 
\end{equation}
Additionally we will use from Eq. (\ref{PNTbnd}) that for $z > 3$, 
 $\frac{1}{\pi(z)} > \frac{2\log z}{3z} >\frac{1}{3z}$. 
 
 Now, putting $x = q^{m-1}$, $y=q^{m-2}$ and $z = q^{m+1}$, it follows from the previous observations that
 \begin{eqnarray*}
s_n &\ge& \frac{\log(\pi(q^{m-1})-\pi(q^{m-2}))}{\log\pi(q^{m+1})}
\ge \frac{\log\left(\frac{q^{m-1}(q-1)}{\log q^{m-1}} \right) + \log((m-2)\log q)}{\log \pi(q^{m+1})} \\
&\ge& \frac{(m-1)\log q + \log(q-1) - \log((m-1)\log q) + \log((m-2)\log q)}{(m+1)\log q + \log 3}\\
&\thicksim& \frac{m-1}{m+1}
\end{eqnarray*}
As  $n$ grows,  $m$ also grows and
$\ \displaystyle \frac{m-1}{m+1} \to 1\ $.
Therefore $\liminf_{n \to +\infty} s_n = 1$. 
We then have that 
\[ \varepsilon({\cal H}(\{q^n\})) = \lim_{n \to \infty} \frac{\log \pi_{\cal H}(q^n)}{\log\pi(q^n)} = 1 .\]

\medskip
We now show the existence of a sentence in ${\cal R}ing(0,+,*,<) + MOD+M$ 
such that its spectrum
has no exponential density. Here we need the logical description of 
$TC^0 = {\cal R}ing(0,+,*,<) + MOD+M$ via FO(COUNT) (see Theorem \ref{focount2tc}).

 \begin{theorem}\label{thm:nonExp4R} There exists a sentence 
 $\theta \in {\cal R}ing(0,+,*,<) + MOD+M$ such that its spectrum $Sp(\theta)$ 
 has no exponential density.
\end{theorem}
\begin{proof}
Fix a prime $q \ge 7$ and consider the sequence $s_n = q^{q^n}$, $n=2,3,\ldots$.
We need to find $r > 3$ such that 
$r\log \pi(s_n) < \log \pi(s_{n+1})$. 
Let $4<a<q$. Then, using Eq. (\ref{PNTbnd}), 
\begin{eqnarray*} a\log\pi(q^{q^{n}}) &< & 
a\log\left(\frac{3}{2} \frac{q^{q^{n}}}{\log q^{q^{n}}}\right) =
\log \left(\frac{3}{2}\right)^a + \log\left(\frac{q^{aq^n}}{q^{an}\log^a q} \right) \\
&<&
\log \left(\frac{3}{2}\right)^a + \log\left(\frac{q^{qq^{n}}}{q^{n+1}\log q}\right) 
< \log \left(\frac{3}{2}\right)^a + \log 2\pi(q^{q^{n+1}})
\end{eqnarray*}
Therefore, for $n$ large we have 
\begin{equation*}
(a-1)\log\pi(q^{q^{n}}) < \left(a- \frac{\log 2\left(\frac{3}{2}\right)^a}{\log\pi(q^{q^{n}})} \right)\log\pi(q^{q^{n}}) 
< \log\pi(q^{q^{n+1}})
\end{equation*}
and we have that the sequence $\{q^{q^n}\}$ is $\log$-thin for $r=a-1 >3$.
Now, by Theorem \ref{thm:Nodense}, the alternating set
\begin{equation}\label{setqqn} 
{\cal H}(\{q^{q^n}\})= \mathbb{P} \cap \left( (q^{q^2}, q^{q^3}) \cup (q^{q^4}, q^{q^5}) \cup \cdots \cup (q^{q^{2k}}, q^{q^{2k+1}}) \cup \cdots \right)
\end{equation}
has no exponential density. 

We need to define a sentence $\theta$ whose spectrum is the set 
${\cal H}(\{q^{q^n}\})$.
The sentence $\theta$ will be such that 
for any structure  $\mathbb{Z}_p \models \theta$ we have that 
\begin{quote}
the size of the model is a prime $p$ with  $q^{q^{2n}}<p<q^{q^{2n+1}}$, for some natural $n$.
\end{quote}
We begin by producing a formula $SUPEXP_q(z)$ in FO(COUNT) that expresses 
``$z= q^{q^n}$ for some $n>0$". Formally, this is equivalent to saying that 
\begin{quote}
{\em ``$z$ is a power of $q$ and the number of powers of $q$ below $z$ is a power of $q$"} 
\end{quote}
The first clause of the above conjunction have been already defined in the logic 
${\cal R}ing(0,+,*,<)$ by the formula $EXP_{q}(z)$ in (\ref{fmlEXP}). To logically define  the second clause we need the counting quantifiers. The required formula is the following:
\begin{eqnarray*}
SUPEXP_q(z) &:=& EXP_q(z) \wedge (\exists i)[(\exists! i y)(y<z \wedge EXP_q(y))\\ 
&\wedge &
\forall j (j\not=1 \wedge \exists k (k\cdot j = i) \to \exists h (q\cdot h = j))]
\end{eqnarray*}

If we additionally require that the number $i$ (of powers of $q$ below $z$) is a square, we have a logical expression for $z=q^{q^{2n}}$. The formula is the following:
\begin{eqnarray*}
SUPEXP_{q^2}(z) &:=& EXP_q(z) \wedge (\exists i)[(\exists! i y) (y<z \wedge EXP_q(y))\\ 
&\wedge &
\forall j (j\not=1 \wedge \exists k (k\cdot j = i) \to \exists h (q\cdot h = j)) \wedge \exists k (k\cdot k = i)]
\end{eqnarray*}
 
 We now express the properties:
 \begin{itemize}
\item {\em ``$z$ is the maximal power of the form $q^{q^n}$"}
\item {\em ``$z$ is the maximal power of the form $q^{q^{2n}}$"}
\end{itemize}
with formulas $SUPMAXEXP_{q}(z)$ and $SUPMAXEXP_{q^2}(z)$, which are
similar in their form to formulas $MAXEXP_{q}(z)$ and $MAXEXP_{q^2}(z)$ in 
(\ref{fmlMEXP}) and (\ref{fmlMEXP2}), respectively, but using our new formulas
$SUPEXP_{q}(z)$ and $SUPEXP_{q^2}(z)$.
Then, the required sentence $\theta$ is 
\[
\theta:=PRIME \wedge \exists z\left(SUPMAXEXP_{q}(z)\wedge SUPMAXEXP_{q^{2}}(z)\right) 
\]
This sentence belongs to FO(COUNT), hence to ${\cal R}ing(0,+,*,<) + MOD+M$, 
and its spectrum has no exponential density.
\end{proof}
 
\section{Conclusions}\label{sec:9}

We have established the separation of subclasses of ${\cal R}ing(0,+,*,<)+MOD+M$ using results from number theory, the notion of prime spectra of sentences and analyzing their natural and exponential densities.
We believe that the algebraic methodology employed is interesting in its own right and should be further exploited. 
%
%
  Of particular interest to us are the following  questions:
\begin{itemize}
\item Does every spectrum  in ${\cal R}ing(0,+,*,<)$ has a exponential density? If that is the case, then this would separate this logic
from  ${\cal R}ing(0,+,*,<)+MOD+M$. 
This would constitute a proof from Class Field Theory perspective of the 
known result 
$$DLOGTIME\mbox{-uniform }AC^0 \not=DLOGTIME\mbox{-uniform }TC^0$$ 
\item What can be said of the spectrum of a sentence in ${\cal R}ing(0,+,*)+MOD(n)$ for $n$ a positive integer? 
 The goal here is to separate ${\cal R}ing(0,+,*)+ MOD(n)$ from ${\cal R}ing(0,+,*)+ MOD(m)$, for $m\neq n$ positive integers. The same when the built-in order is present, which will yield separations among the $ACC(n)$ classes improving
Smolensky results  \cite{Smo87} to the uniform setting.
\item Is ${\cal R}ing(0,+,*) \not= {\cal R}ing(0,+,*)+ MOD(2)$?
\item What can be said of the spectrum of a sentence in 
${\cal R}ing(0,+,*,<)+ MOD + M$? 
We expect these sets to be much more untamed  than the spectra of sentences in ${\cal R}ing(0,+,*,<)$ because of the expressive power of the majority quantifier. 
Thus, possibly a refined version of the exponential density is necessary to  study these spectra.
\end{itemize}

\end{document}